\documentclass[letterpaper, 10 pt, conference]{ieeeconf}  

\IEEEoverridecommandlockouts                              

\pdfoutput=1
\usepackage{graphicx} 
\usepackage{amsmath} 
\usepackage{amssymb}  
\usepackage{color}
\usepackage[noadjust]{cite}

\newcommand{\A}{\mathcal{A}}
\newcommand{\B}{\mathcal{B}}
\newcommand{\C}{\mathcal{C}}
\newcommand{\D}{\mathcal{D}}
\newcommand{\R}{\mathbb{R}}
\newcommand{\Sb}{\mathbb{S}}
\newcommand{\RL}{\mathbb{RL}}
\newcommand{\RH}{\mathbb{RH}}
\newcommand{\Fc}{\mathcal{F}}
\newcommand{\Mc}{\mathcal{M}}

\newcommand{\Lc}{\mathcal{L}}

\DeclareMathOperator{\col}{col}
\DeclareMathOperator{\diag}{diag}

\DeclareMathOperator{\sgn}{sgn}

\usepackage{amsthm}

\theoremstyle{plain}
\newtheorem{thm}{Theorem}
\newtheorem{prop}[thm]{Proposition}
\newtheorem{cor}[thm]{Corollary}

\theoremstyle{definition}
\newtheorem{dfn}{Definition}

\newtheorem{asmp}{Assumption}

\theoremstyle{remark}

\title{\LARGE \bf
Robust Contraction Analysis of Nonlinear Systems via Differential IQC
}

\author{Ruigang Wang and Ian R. Manchester
\thanks{This work was supported by the Australian Research Council.}
\thanks{The authors are with the Australian Centre for Field Robotics, The University of Sydney, Sydney, NSW 2006, Australia
        (e-mail: {\tt\small  ian.manchester@sydney.edu.au}).}%
}

\begin{document}

\maketitle
\thispagestyle{empty}
\pagestyle{empty}

\begin{abstract}
	We present a new approach to verifying contraction and $L_2$-gain of uncertain nonlinear systems, extending the well-known method of integral quadratic constraints. The uncertain system consists of a feedback interconnection of a nonlinear nominal system and uncertainties satisfying differential integral quadratic constraints. 
	 A pointwise linear matrix inequality condition is formulated to verify the closed-loop differential $ L_2 $ gain, which can lead to global reference-independent $ L_2 $ gain performance of the nonlinear uncertain system. For a polynomial nominal system, the convex verification conditions can be solved via sum-of-squares programming. A simple computational example based on jet-engine surge with input delays illustrates the approach.
\end{abstract}

\section{Introduction}
The integral quadratic constraint (IQC) approach of \cite{Megretski:1997} provides a flexible framework for robustness analysis of uncertain systems, and includes as special cases many previously-proposed methods. The basic setup of \cite{Megretski:1997} is an interconnection of a nominal system -- which is a stable, finite-dimensional linear time-invariant system -- and uncertainties that are known to satisfy integral quadratic constraints. The ``uncertainties'' may be known but ``troublesome" components (e.g., time delay or other infinite-dimensional dynamics, saturation or other nonlinearties) or an unknown but bounded dynamics dynamics. Closed-loop stability and performance can be verified by solving a linear matrix inequality (semidefinite program).

Advances in computational analysis of nonlinear systems, especially polynomial systems via sum-of-squares programming \cite{Parrilo:2003}, motivate extending the IQC approach to scenarios in which the ``nominal'' system is nonlinear, but still amenable to computation, e.g. described by a relatively low-degree polynomial vector field. However, many of the most powerful IQCs used in \cite{Megretski:1997} are ``soft'' IQCs, meaning that their time-domain representations do not necessarily hold for all finite times (so-called ``hard'' IQCs), as would be required for standard analysis methods of nonlinear systems \cite{Schaft:2017}.

 Recently, the work \cite{Seiler:2015} proved that, under rather mild conditions, most soft IQCs do in fact have hard representations, in particular if the associated frequency-domain multiplier admits a $ J $-factorization. Related results were obtained in \cite{veenman2013stability, Scherer:2018}. Based on this, several approaches have been developed to relax the assumption on nominal models. In \cite{Pfifer:2015,Wang:2016}, the IQC-based analysis framework has been applied to linear parameter-varying (LPV) systems. In \cite{Carrasco:2018}, a time-domain IQC theorem using graph separation was proposed for the feedback interconnection of two nonlinear systems. The stability condition was stated purely on input-output relations without involving linear matrix inequalities (LMIs).

The stability and performance notation used in the IQC framework is based on $ L_2 $ gain with respect to one preferred equilibrium (e.g., the origin). However, for many nonlinear systems, there are reasons to prefer a stronger stability analysis tool which is independent of the reference \cite{Wang:2017}. These cases often require the input-output stability to consider both \emph{boundedness} (i.e., bounded inputs produce bounded outputs) and \emph{continuity} (the outputs are not critically sensitive to small changes in inputs). Although both continuity and boundedness of input-output stability was introduced by Zames \cite{Zames:1966}, research on boundedness of solutions and stability of special solutions (e.g., a known equilibrium) have dominated existing literature. In \cite{Desoer:1975}, Desoer and Vidyasagar show that bounded and continuous solutions with respect to inputs can be implied by incremental $ L_2 $ stability while only boundedness can be ensured by the standard $ L_2 $-gain stability. The incremental IQC was briefly discussed in \cite{Megretski:1997} and later on applied to the harmonic analysis of uncertain systems in \cite{Rantzer:1997, Jonsson:2001}. However, it was proved in \cite{Kulkarni:2002} that stability multipliers such as Zames-Falb, Popov, and RL/RC multipliers are not directly applicable for incremental stability analysis as they cannot preserve incremental positivity. Recently, a weaker notation -- \emph{equilibrium-independent} IQC was introduced to describe nonlinear systems \cite{Summers:2013}. The frequency-domain IQC theorem is then applied for robustness analysis of networked passive systems with delays.  

Contraction \cite{Lohmiller:1998} is a strong form of stability meaning that, roughly speaking, all solutions converge to each other exponentially. The underlying idea is to investigate the local stability of the linearized system (differential dynamics) along any admissible trajectory, from whichglobal incremental stability can be established via integration of local analysis along certain path. The benefits of contraction analysis are two folds. First, no specific knowledge about the particular trajectory or reference signal is required during the analysis stage, which leads to a \emph{reference-independent} approach. Second, the analysis problem has a convex formulation using Riemannian metric \cite{Aylward:2008}, which allows for a numerically efficient optimization method -- sum-of-squares (SOS) programming \cite{Parrilo:2003}. Recently, a general differential Lyapunov framework based on Finsler metric was presented in \cite{Forni:2014}. The extension of differential Lyapunov theory to input/output dynamics -- differential dissipativity \cite{Forni:2013,Schaft:2013} was applied to system identification \cite{Tobenkin:2010}, robust analysis of a limit cycle \cite{Manchester:2014} and distributed control of chemical systems \cite{Wang:2017b}. A controller synthesis and realization framework based on control contraction metric (CCM) was developed in \cite{Manchester:2017}. It was then extended to performance design in \cite{Manchester:2018}. In \cite{Wang:2006}, contraction analysis based on certain metrics was applied to group cooperation subject to time-delayed communications. However, there are as yet few precise results on contraction of uncertain systems containing a broad range of ``troublesome'' components.  

In this paper, we present an approach to verify robust contraction  and performance ($L_2$ gain) of an uncertain nonlinear system by lifting the time-domain IQC theorem into the differential setting. The uncertain system consists of a feedback interconnection of a nominal nonlinear model and a perturbation. Here we assume that the nominal model belongs to a class of ``less troublesome'' nonlinear dynamics whose differential $ L_2 $ gain can be obtained via pointwise LMIs. A novel IQC, namely \emph{differential IQC}, is then introduced to replace the ``troublesome'' uncertainty. By using the results from IQC-based analysis for LPV systems \cite{Pfifer:2015, Pfifer:2016}, we develop a pointwise LMI condition which yields a bounded differential $ L_2 $ gain for the uncertain systems. Finally, we prove that global incremental $ L_2 $-gain performance can be guaranteed if the local analysis result is satisfied for all one-parameter solution families. We also show that a \emph{global} $ L_2 $-gain (a weaker notation compared with incremental one) bound with respect to certain reference trajectory can be inferred if the differential dissipation test is only validated for certain class of solution families.

The paper is structured as follows. Section~\ref{sec:background} presents the background on contraction analysis and IQC. Section~\ref{sec:DIQC} introduces the concept of $ \delta $-IQC. The main results on robust contraction analysis based on differential dissipativity and $ \delta $-IQCs are given in Section~\ref{sec:main}. Section~\ref{sec:example} provides a simple computational example.

\section{Background}\label{sec:background}

\subsection{Notation}

Most notation is from \cite{Zhou:1996}. 
For a matrix $ M\in\C^{m\times n} $, $ M' $ denotes the transpose and $ M^* $ denotes the conjugate transpose, respectively. The para-Hermitian conjugate of $ G\in\RH_\infty $, denoted as $ G^\sim $, is defined by $ G^\sim(s):=G(-s)' $. $ \mathcal{C}^k $ denotes the set of vector signals on $ \R $ which have $ k $th order derivative. $ \mathcal{L}_2 $ is the space of square-integrable vector signals on $ \R_{\geq 0} $, i.e., $ \|f\|:=(\int_{0}^{\infty}|f(t)|dt)^{1/2}<\infty $ where $ |x| $ denotes the Euclidean norm of a vector $ x $. The causal truncation $ (\cdot)_T $ is defined by $ (f)_T(t):=f(t) $ for $ t\leq T $ and 0 otherwise. $ \mathcal{L}_{2e} $ is the space of vector signals on $ \R_{\geq 0} $ whose causal truncation belongs to $ \mathcal{L}_2 $. Let $ ARE(A,B,Q,R,S) $ denote the following Algebraic Riccati Equation (ARE):
\begin{equation}
A'X+XA-(XB+S)R^{-1}(XB+S)'+Q=0.
\end{equation}
The stabilizing solution $ X\in\Sb $, if it exists, is such that $ A-BR^{-1}(XB+S)' $ is Hurwitz.

A Riemannian metric on $ \R^n $ is a symmetric positive definite matrix function $ M(x) $, smooth in $ x $, which defines a inner product $ \langle\delta_1,\delta_2 \rangle_x:=\delta_1'M(x)\delta_2 $ for any two tangent vector $ \delta_1,\delta_2 $. A metric is called \emph{uniformly bounded} if there exist positive constants $ a_2\geq a_1 $ such that $ a_1I\leq M(x)\leq a_2I,\,\forall x\in\R^n $. 
$ \Gamma(x_0,x_1) $ denotes the set of piecewise smooth paths $ c:[0,1]\rightarrow\R^n $ with $ c(0)=x_0 $ and $ c(1)=x_1 $. The curved length of $ c(\cdot) $ is defined by 
\begin{equation}
	\ell(c):=\int_{0}^{1}\sqrt{\langle c_s,c_s \rangle_{c(s)}}ds
\end{equation} 
where $ c_s=\partial c/\partial s $. The \emph{geodesic} $ \gamma(\cdot) $ denotes a path with the minimal length, i.e., $ \ell(\gamma)=\inf_{c\in\Gamma(x_0,x_1)}\ell(c) $. 
For more details see \cite{Do-Carmo:1992}. 

\subsection{Differential $ L_2 $ Gain via Contraction Analysis}
Consider a nonlinear system of the form
\begin{equation}\label{eq:system}
	\dot{x}=f(x,d), \quad e=h(x,d)
\end{equation}
where $ x(t)\in\R^{n_x} $, $ d(t)\in\R^{n_d} $, $ e(t)\in\R^{n_e} $ are the state, disturbance and performance output, respectively. For simplicity, $ f,h $ are assumed to be smooth and time-invariant. 

Instead of investigating stability properties of one solution $ (x,d,e)(\cdot) $, we are interested in an one-parameter solution family defined as follows \cite{Forni:2013}. In what follows we denote $ \rho:=(x,d,e) $. 
\begin{dfn}
	A set of solutions $ \Omega_{\rho}=\{\overline{\rho}(\cdot,s)\}_{s\in\R} $
	where $ \overline{\rho}(\cdot,0)=\rho(\cdot) $ is said to be an \emph{one-parameter solution family} to \eqref{eq:system} if  $ \overline{\rho}(\cdot,\cdot)\in\mathcal{C}^2 $ for almost all $ t\in\R_{\geq 0},\,s\in\R $, and $ \overline{\rho}(\cdot,s) $ satisfies \eqref{eq:system} for all $ s\in\R $. 
\end{dfn}
Compared to the time $ t $, variable $ s $ can be seen as a spatial parameter. The partial derivative $ \frac{\partial \overline{x}}{\partial t} $ at each $ t $ characterizes the time evolution of a curve $ \overline{x}(t,\cdot) $ subject to the nonlinear model \eqref{eq:system}. The partial derivative $ \frac{\partial \overline{x}}{\partial s} $ at each $ s $ characterizes the behavior of tangent vector that moves along the solution $ \overline{x}(\cdot,s) $, as shown in Fig.~\ref{fig:solution-familiy}. This behavior can be modeled by the following \emph{differential dynamics} (\cite{Lohmiller:1998}):
\begin{equation}\label{eq:diff-dynamics}
\begin{split}
\dot{\delta}_x&=A(x,d)\delta_x+B(x,d)\delta_d:=\frac{\partial f}{\partial x}\delta_x+\frac{\partial f}{\partial d}\delta_d \\ \delta_e&=C(x,d)\delta_x+D(x,d)\delta_d:=\frac{\partial h}{\partial x}\delta_x+\frac{\partial h}{\partial d}\delta_d
\end{split}
\end{equation}
where $ (\delta_x,\delta_d,\delta_e)=\left(\frac{\partial \overline{x}}{\partial s},\frac{\partial \overline{d}}{\partial s},\frac{\partial \overline{e}}{\partial s}\right) $. 

\begin{figure}[!bt]
	\centering
	\includegraphics[width=0.55\linewidth]{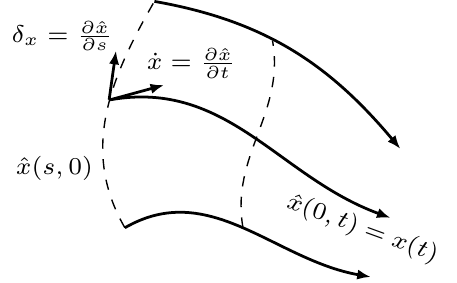}
	\caption{Geometrical interpretation of the differential dynamics based on an one-parameter solution family.}\label{fig:solution-familiy}
\end{figure}

The performance condition for nonlinear systems usually involves a bound on the $ L_2 $ gain from $ d $ to $ e $. There are several different formulations of $ L_2 $ gain. Since the standard $ L_2 $ gain does not consider continuity of input-output stability, this work investigates the following strong notations.
\begin{dfn}
	System \eqref{eq:system} is said to have an \emph{incremental} $ L_2 $-gain bound of $ \alpha>0 $ if for all pair of solutions with initial conditions $ x_0(0),x_1(0) $ and input $ d_0,d_1\in\Lc_{2e} $, and for all $ T>0 $ solutions exist and 
	\begin{equation}\label{eq:global-L2}
	\|(e_1-e_0)_T\|^2\leq \alpha^2\|(d_1-d_0)_T\|^2+b(x_0(0),x_1(0))
	\end{equation}
	for some function $ b(x_0,x_1)\geq 0 $ with $ b(x,x)=0 $.
\end{dfn}
\begin{dfn}
	System \eqref{eq:system} is said to have a \emph{global} $ L_2 $-gain bound of $ \alpha>0 $ if there exists a unique solution $ \rho^*(\cdot) $, any initial condition $ x(0) $ and input $ d-d^*\in\Lc_{2e} $, and for all $ T>0 $ solutions exist and 
	\begin{equation}
	\|(e-e^*)_T\|^2\leq \alpha^2\|(d-d^*)_T\|^2+b(x(0),x^*(0))
	\end{equation}
	for some function $ b(x_0,x_1)\geq 0 $ with $ b(x,x)=0 $. 
\end{dfn}
Note that global $ L_2 $ gain stronger than the equilibrium-independent counterpart \cite{Summers:2013}, but weaker than incremental one, since it only requires the convergence between a particular (i.e., system solution and reference trajectory) rather than arbitrary pair of trajectories, as shown in Fig.~\ref{fig:stability}. 

\begin{figure}[!bt]
	\centering
	\begin{tabular}{cc}
		\includegraphics[width=0.4\linewidth]{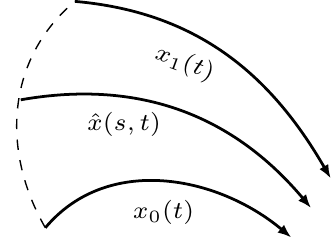} &
		\includegraphics[width=0.5\linewidth]{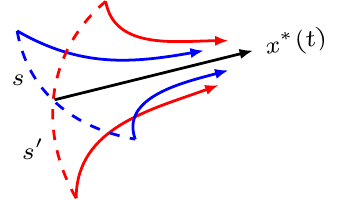} \\
		(a) incremental & (b) global
	\end{tabular}
	\caption{Different notations for reference-independent stability.}\label{fig:stability}
\end{figure}

We also work on the \emph{differential} $ L_2 $ gain, that is, system \eqref{eq:system} is said to have a differential $ L_2 $-gain bound of $ \alpha $ if for all $ T>0 $ 
\begin{equation}
	\|(\delta_e)_T\|^2\leq \alpha^2\|(\delta_d)_T\|^2+\beta(x(0),\delta_x(0))
\end{equation}
where $ \beta(x,\delta_x)\geq 0 $ with $ \beta(x,0)=0 $ for all $ x $. Applying standard results \cite[Th. 3.1.11]{Schaft:2017} to the joint system  \eqref{eq:system}, \eqref{eq:diff-dynamics} gives that a sufficient, and in some cases necessary, condition for the differential $ L_2 $-gain bound is the existence of a differential storage function $ V(x,\delta_x)\geq 0 $ with $ V(x,0)=0 $ such that
\begin{equation}\label{eq:diff-dissipation}
	\dot{V}(x,\delta_x)\leq \alpha^2|\delta_d|^2-|\delta_e|^2.
\end{equation}
Here we are interested in the differential storage function induced by a Riemannian metric, i.e., $ V(x,\delta_x)=\delta_x'\Mc(x)\delta_x $. Other  more general metrics are possible, see \cite{Forni:2013,Chaffey:2018}, however we focus on the Riemannian case as the analysis problem can have a convex formulation. 


\subsection{Integral Quadratic Constraints}
In the IQC framework, the uncertainty (either a known but ``difficult'' component or  unknown dynamics) is described by an operator $ \varDelta $, which refers to an input-output system 
\begin{equation}
	w=\varDelta(v).
\end{equation} 
An operator $ \varDelta $ is said to be \emph{casual} if $ \varDelta_T(v):=(w)_T=\varDelta((v)_T) $ for all $ T\geq 0 $. It is said to be \emph{bounded} if there exists $ C_1 $ such that $ \|\varDelta_T(v)\|\leq C_1\|(v)_T\| $ for all $ T>0 $ and $v\in \mathcal{L}_{2e} $.  The underlying idea of IQC analysis is to replace $ \varDelta $ with a constraint or set of constraints that it is known to satisfy, and which are convenient for analysis that is convenient for analysis via LMIs. In particular, the constraints are frequency-weighted integral inequalities of the form below.

\begin{dfn}\label{dfn:iqc-frequency}
	Let $ \Pi=\Pi^\sim\in\RL_{\infty}^{(n_v+n_w)\times(n_v+n_w)} $ be given. Two signals $ v\in \Lc_2^{n_v} $ and $ w\in \Lc_2^{n_w} $ satisfy the \emph{frequency domain IQC} defined by the multiplier $ \Pi $ if 
	\begin{equation}\label{eq:iqc-frequency}
	\int_{-\infty}^{\infty}
	\begin{bmatrix}
	\widehat{V}(j\omega) \\ \widehat{W}(j\omega)
	\end{bmatrix}^*\Pi(j\omega)
	\begin{bmatrix}
	\widehat{V}(j\omega) \\ \widehat{W}(j\omega)
	\end{bmatrix}d\omega \geq 0
	\end{equation}
	where $ \widehat{V} $ and $ \widehat{W} $ are Fourier transforms of $ v $ and $ w $. A bounded, causal operator $ \varDelta:\Lc_{2e}^{n_v}\rightarrow \Lc_{2e}^{n_w} $ satisfies the frequency domain IQC defined by $ \Pi $ if \eqref{eq:iqc-frequency} holds for all $ v\in \Lc_2^{n_v} $ and $ w=\varDelta(v) $.
\end{dfn}

IQCs can also be expressed in the time-domain via Parseval's theorem, and can be interpreted as comparing filtered versions of the input and output of $\Delta$ as in Fig. \ref{fig:iqc}.
\begin{dfn}\label{dfn:iqc-time}
	Let $ \Psi\in\RH_\infty^{n_z\times(n_v+n_w)} $ and $ M=M' $. Two signals $ v\in \Lc_{2}^{n_v},\,w\in \Lc_{2}^{n_w} $ satisfy the \emph{time domain IQC} defined by the multiplier $ \Psi $ and matrix $ M $ if the following inequality holds 
	\begin{equation}\label{eq:iqc-time}
	\int_{0}^{\infty}z'(t)Mz(t)dt\geq 0
	\end{equation}
	where $ z $ is the output of $ \Psi $ with the zero initial condition and inputs $ (v,w) $. A bounded, causal operator $ \varDelta $ satisfies the time domain IQC defined by $ (\Psi,M) $ if \eqref{eq:iqc-time} holds for all $ v\in \Lc_{2}^{n_v} $ and $ w=\varDelta(v) $.
\end{dfn}
The time domain IQC in \eqref{eq:iqc-time} is referred as a \emph{soft IQC} in \cite{Megretski:1997}, and it is important to note that it assumes that all signals are in $L^2$. If  the time domain constraint $ \int_{0}^{T}z'(t)Mz(t)dt\geq 0 $ holds for all $ T\geq 0 $, this is called a \emph{hard IQC}. The hard/soft property is not strictly inherent to the multiplier $ \Pi $ but instead depends on the non-unique factorization $ \Pi=\Psi^{\sim}M\Psi $. In particulary, a key additional assumptions on the multiplier $ \Pi $ is that it admits a J-spectral factorization, in which case a hard factorization can be constructed \cite{Seiler:2015}.

\begin{figure}[!bt]
	\centering
	\includegraphics[width=0.55\linewidth]{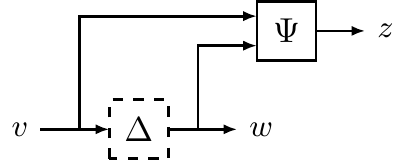}
	\caption{Graphical interpretation of the IQC.}\label{fig:iqc}
\end{figure}

\section{Differential IQC}\label{sec:DIQC}
The conventional IQC cannot be directly applied to contraction analysis since it is defined with respect to one preferred equilibrium. In this section, we will introduce the concept of differential IQC, which is used to replace the ``troublesome'' perturbation. 

An operator $ \varDelta $ is said to be \emph{locally affine bounded} if there exists $ C_0,\,C_1\geq 0 $ such that 
\begin{equation}
	\|\varDelta_T(v_1)-\varDelta_T(v_2)\|\leq C_0+C_1\|(v_1-v_2)_T\|
\end{equation}
for all $ T>0 $ and $v_i\in \Lc_{2e} $. It is called \emph{locally bounded} if this holds with $ C_0=0 $.
For a locally bounded operator $ \varDelta $, its associated differential operator $ \partial\varDelta $ can be defined by
\begin{equation}
	\delta_w=\partial\varDelta(v;\delta_v):=\limsup_{s\rightarrow 0^+}\frac{\varDelta(v+s\delta_v)-\varDelta(v)}{s}
\end{equation}
where $ v,\delta_v\in  \Lc_{2e} $. 

The differential IQC is an integral quadratic constraint specified on $ (\delta_v,\delta_w) $. The following definition is given in frequency domain (the time-domain definition is omitted due to space restrictions, but it is similar to Definition~\ref{dfn:iqc-time}).
\begin{dfn}
	An operator $ \varDelta: \Lc_{2e}^{n_v}\rightarrow \Lc_{2e}^{n_w} $ is said to satisfy the frequency-domain \emph{differential IQC} ($ \delta $-IQC) defined by the multiplier $ \Pi=\Pi^\sim\in\RL_{\infty}^{(n_v+n_w)\times(n_v+n_w)} $ if $ \varDelta $ is casual and locally bounded, and the following inequality holds
	\begin{equation}\label{eq:iqc-diff}
	\int_{-\infty}^{\infty}
	\begin{bmatrix}
	\hat{\delta}_v(j\omega) \\ \hat{\delta}_w(j\omega)
	\end{bmatrix}^*\Pi(j\omega)
	\begin{bmatrix}
	\hat{\delta}_v(j\omega) \\ \hat{\delta}_w(j\omega)
	\end{bmatrix}d\omega \geq 0
	\end{equation}
	where $ \hat{\delta}_v $ and $ \hat{\delta}_w $ are Fourier transforms of $ \delta_v \in \Lc_2^{n_v} $ and $ \delta_w=\partial\varDelta(v;\delta_v) \in \Lc_2^{n_w} $.
\end{dfn}
A similar definition can be given for incremental IQCs.
\begin{dfn}
	A causal operator $ \varDelta $ satisfies the frequency-domain \emph{incremental IQC}  defined by $ \Pi $ if for any pair of input-output trajectories $ (v_0,w_0)(\cdot) $ and $ (v_1,w_1)(\cdot) $ with $ d_v=v_1-v_0\in \Lc_2^{n_v} $ and $ d_w=\varDelta(v_1)-\varDelta(v_0)\in \Lc_2^{n_w} $, the following inequality holds
	\begin{equation}\label{eq:iqc-incremental}
	\int_{-\infty}^{\infty}
	\begin{bmatrix}
	\hat{d}_v(j\omega) \\ \hat{d}_w(j\omega)
	\end{bmatrix}^*\Pi(j\omega)
	\begin{bmatrix}
	\hat{d}_v(j\omega) \\ \hat{d}_w(j\omega)
	\end{bmatrix}d\omega \geq 0
	\end{equation}
	where $ \hat{d}_v $ and $ \hat{d}_w $ are Fourier transforms of $ d_v $ and $ d_w $. If the above condition only holds for a particular reference input (i.e.,$ v_0=v^* $), the operator $ \varDelta $ satisfies an \emph{global IQC} defined by $ \Pi $.
\end{dfn}
Note that a global IQC is just a standard IQC \cite{Megretski:1997} with respect to a trajectory which is not necessarily at the origin.

For locally bounded operators, $ \delta $-IQC and incremental IQCs are equivalent under a mild assumption on the multiplier, which is also important for hard factorization.
\begin{asmp}\label{asmp:1}
	The multiplier $ \Pi $ satisfies $ \Pi_{vv}(j\omega)\geq 0 $ and $ \Pi_{ww}(j\omega)\leq 0 $ for all $ \omega\in\R\cup\{\infty\} $, where $ \begin{bmatrix}
	\Pi_{vv} & \Pi_{vw} \\ \Pi_{vw}^\sim & \Pi_{ww}
	\end{bmatrix} $ with $ \Pi_{vv}\in\RL_{\infty}^{n_v\times n_v} $ is a partition of the multiplier $ \Pi $.
\end{asmp}

\begin{prop}\label{prop:1}
	Assume that Assumption~\ref{asmp:1} holds for the multiplier $ \Pi $. The locally bounded operator $ \varDelta $ satisfies the $ \delta $-IQC induced by $ \Pi $ for all one-parameter solution families $ \Omega_{\varrho} $ with $ \varrho:=(v,w) $ if and only if it satisfies the incremental IQC induced by $ \Pi $.
\end{prop}
\begin{proof}[Proof of Proposition~\ref{prop:1}]
	``{\bf if}": For any $ v\in \Lc_{2e}^{n_v} $ and $ \delta_v\in \Lc_2^{n_v} $, we take $ v_0=v,w_0=w $ and $ v_1=v+s\delta_v $, $ w_1=\varDelta(v_1) $. For a sufficiently small $ s $, we have $ d_w=s\delta_w=s\partial \varDelta(v;\delta_v) $. Since $ \varDelta $ is locally bounded, there exists a set-valued map $ D_\varDelta $ such that $ \partial\varDelta(v;\delta_v)= D_\varDelta(v)\delta_v $. Condition \eqref{eq:iqc-diff} follows by substituting $ d_v=s\delta_v $ and $ d_w=sD_\varDelta(v)\delta_v $ into \eqref{eq:iqc-incremental}.
	
	``{\bf only if}": For any pair of input-output behaviors $ (v_0,w_0)$ and $(v_1,w_1) $ satisfying $ d_v, d_w\in \Lc_2  $, by taking the parameterization $ v(s)=(1-s)v_0+sv_1 $ and $ w(s)=\varDelta(v(s)) $, we have $ \delta_v(s)=d_v\in \Lc_2 $ and $ \delta_w(s)=\partial\varDelta(v(s),\delta_v(s))\in \Lc_2 $ as $ \partial \varDelta $ is bounded. Integration of \eqref{eq:iqc-diff} over $ s\in[0,1] $ yields (the dependence on $ j\omega $ is omitted for simplicity):
	\begin{equation}\label{eq:delta-Delta}
	\begin{split}
	0\leq& \int_{0}^{1}\int_{-\infty}^{\infty}
	\begin{bmatrix}
	\hat{\delta}_v(s) \\ \hat{\delta}_w(s)
	\end{bmatrix}^*\Pi
	\begin{bmatrix}
	\hat{\delta}_v(s) \\ \hat{\delta}_w(s)
	\end{bmatrix}d\omega ds \\
	=&\int_{-\infty}^{\infty}\int_{0}^{1}
	\begin{bmatrix}
	\hat{d}_v \\ \hat{\delta}_w(s)
	\end{bmatrix}^*
	\begin{bmatrix}
	\Pi_{vv} & \Pi_{vw} \\ \Pi_{vw}^\sim & \Pi_{ww}
	\end{bmatrix}
	\begin{bmatrix}
	\hat{d}_v \\ \hat{\delta}_w(s)
	\end{bmatrix}ds d\omega \\
	=&\int_{-\infty}^{\infty}
	\begin{bmatrix}
	\hat{d}_v \\ \hat{d}_w
	\end{bmatrix}^*
	\begin{bmatrix}
	\Pi_{vv} & \Pi_{vw} \\ \Pi_{vw}^\sim & 0
	\end{bmatrix}
	\begin{bmatrix}
	\hat{d}_v \\ \hat{d}_w
	\end{bmatrix}ds d\omega +\\
	&\; \int_{-\infty}^{\infty}\int_{0}^{1}\hat{\delta}_w^*(s)\Pi_{ww}\hat{\delta}_w(s)dsd\omega.
	\end{split}
	\end{equation}
	Since $ \Pi_{ww}(j\omega)=\Pi_{ww}^\sim(j\omega)\leq 0 $, it yields a factorization $ \Pi_{ww}=-\Lambda_w^\sim\Lambda_w $. From Cauchy-Schwarz inequality, we have
	\begin{equation}\label{eq:delta-Delta-2}
	\begin{split}
	&\int_{0}^{1}\hat{\delta}_w^*(s)\Pi_{ww}\hat{\delta}_w(s)ds=-\int_{0}^{1}|\Lambda_w\hat{\delta}_w(s)|^2ds \\
	&\quad \leq -\left|\int_{0}^{1}\Lambda_w\hat{\delta}_w(s)ds\right|^2
	=\hat{d}_w^*\Pi_{ww}\hat{d}_w.
	\end{split}
	\end{equation}
	The incremental IQC condition \eqref{eq:iqc-incremental} follows from \eqref{eq:delta-Delta}-\eqref{eq:delta-Delta-2}.
\end{proof}
From the above proposition, we have following corollaries.
\begin{cor}\label{coro:2}
	If the operator $ \varDelta $ is linear, then the $ \delta $-IQC induced by $ \Pi $ is equivalent to the IQC induced by $ \Pi $.
\end{cor}
\begin{cor}
	If the operator $ \varDelta $ satisfies the $ \delta $-IQC induced by $ \Pi $ for all $ \Omega_{\varrho^*} $ with a fixed $ \varrho^* $, then it satisfies the global IQC induced by $ \Pi $.
\end{cor}

For general locally affine bounded operators, the corresponding differential operators may not be well-defined since even a small input may produce large output (e.g., the relay operator $ w=\sgn(v) $), therefore, one cannot find any $ \delta $-IQC. However, it may be possible to construct artificial feedback loops encapsulating those operators \cite{Rantzer:1997} which enable $\delta$-IQC analysis. For example, consider the following uncertain system
\begin{equation}\label{eq:encapsulation}
		\dot{y}=-ay-\varDelta_f(y)+v 
\end{equation}
where $ \varDelta_f $ is a viscous friction operator defined by $ \varDelta_f(y)=\sgn(y)(b|y|+c) $ with $ a,b, c>0 $.
The operator $ w=\varDelta_f(y) $ is not locally bounded. However, a locally bounded operator can be constructed via the feedback encapsulation shown in Fig.~\ref{fig:encapsulation}, for which we can find a $\delta$-IQC.
\begin{prop}
	The operator $\varDelta:v\mapsto y$ defined in \eqref{eq:encapsulation} satisfies a differential $L_2$ gain bound of $ \frac{1}{a+b}$.
\end{prop}
\begin{proof}
	We sketch a proof based on the regularisation approach of \cite{fiore2016contraction}. The operator $ \varDelta:v\mapsto y $ is locally bounded. The following smooth system
$
		\dot{y}=-(a+b)y-c\tanh(y/\epsilon)+v
$	with $ \epsilon>0 $ approaches \eqref{eq:system} when $ \epsilon\rightarrow 0 $. The differential dynamics of this approximation are
	\begin{equation}
		\dot{\delta}_y=-(a+b)\delta_y-c/\epsilon(1-\tanh^2(y/\epsilon))\delta_y+\delta_v.
	\end{equation}
	Since $ 1-\tanh^2(y/\epsilon)\geq 0 $, the above differential dynamics has a $ L_2 $-gain bound of $ \frac{1}{a+b} $ for all $ \epsilon>0 $. The claim follows by taking $ \epsilon\rightarrow 0 $ and the results of \cite{fiore2016contraction}.
\end{proof}

\begin{figure}[!bt]
	\centering
	\includegraphics[width=0.7\linewidth]{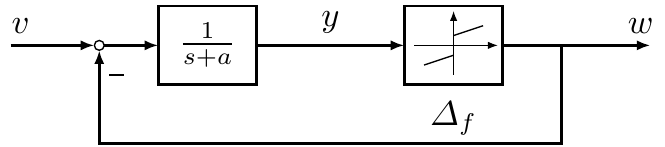}
	\caption{Bounded feedback encapsulation of a viscous friction.}\label{fig:encapsulation}
\end{figure}

\section{Robust Contraction Analysis}\label{sec:main}
The main results of this paper prove robustness results for a nominal nonlinear system in feedback with uncertainties satisfying differential IQCs.

\subsection{Problem Statement}
We consider a robust contraction analysis problem for uncertain nonlinear systems. Here the perturbed system is described by the feedback interconnection of a nominal nonlinear system $ G $ and an uncertainty $ \varDelta $ as shown in Fig.~\ref{fig:feedback}. This feedback interconnection with $ \varDelta $ wrapped around the top of $ G $ is denoted $ \Fc_u(G,\varDelta) $. The nominal system $ G $ is represented by
\begin{equation}\label{eq:nominal-sys}
	\dot{x}=f(x,w,d),\; v=g(x,w,d),\; e=h(x,w,d)
\end{equation}
where $ x(t)\in\R^{n_x} $ is the state, $ w(t)\in\R^{n_w} $ and $ d(t)\in\R^{n_d} $ are inputs, and $ v(t)\in\R^{n_v} $ and $ e(t)\in\R^{n_e} $ are outputs.  For simplicity, $ f,g,h $ are assumed to be smooth and time-invariant. Moreover, the following assumptions are made regarding $ G $ and $ \varDelta $.
\begin{figure}[!bt]
	\centering
	\includegraphics[width=0.45\linewidth]{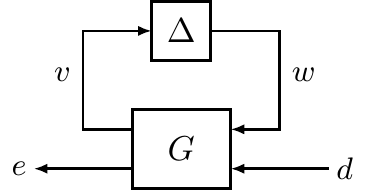}
	\caption{Feedback interconnection}\label{fig:feedback}
\end{figure}
\begin{asmp}\label{asmp:2}
	The nominal nonlinear system $ G $ has a differential $ L_2 $-gain bound from $ \col(w,d) $ to $ \col(v,e) $.
\end{asmp}
\begin{asmp}\label{asmp:3}
	The uncertainty $ \varDelta $ satisfies a collection of frequency domain $ \delta $-IQCs defined by $ \{\Pi_k\}_{1\leq k\leq N} $ with $ \Pi_k\in \RL_{\infty}^{(n_v+n_w)\times(n_v+n_w)} $ satisfying Assumption~\ref{asmp:1}, denoted by $ \partial \varDelta\in \partial\mathbf{\Delta}(\Pi_1,\ldots,\Pi_N) $.
\end{asmp}

\begin{asmp}\label{asmp:4}
	The overall uncertainty satisfies a differential $L^2$ gain bound, and has been normalized to satisfy $ \|\partial \varDelta\|\leq 1 $, so the first $ \delta $-IQC is defined by the multiplier $ \Pi_1=\diag(I_{n_v},-I_{n_w}) $.
\end{asmp}

All these assumptions are used to simplify the algorithm. From Assumption~\ref{asmp:2}, the nominal system is a ``less troublesome'' nonlinear dynamics since the differential $ L_2 $-gain implies bounded and continuous outputs with respect to inputs. The $ \delta $-IQCs in Assumption~\ref{asmp:3} are used to bounded the differential input/output behavior of the perturbation $ \varDelta $. Assumption~\ref{asmp:1} and \ref{asmp:4} are used to ensure that a ``combined'' multiplier $ \Pi_{\lambda}=\sum_{k=1}^{N}\lambda_k\Pi_k $ is a hard IQC and has a $ J $-spectral factorization for all $ \lambda\in\Lambda:=\{\lambda\in\R^{N}\mid \lambda_1>0,\,\lambda_k\geq 0,\, 2\leq k\leq N\} $.  Similar assumptions are made in \cite{Pfifer:2016} for LPV robustness analysis.


\subsection{Robust Performance Condition}
In this section, a pointwise LMI is developed to verify the differential dissipation condition \eqref{eq:diff-dissipation} for the uncertain system $ \Fc_u(G,\varDelta) $. First, the differential dynamics of the nominal system $ G $ can be represented by
\begin{equation}\label{eq:delta-G}
	\begin{bmatrix}
		\dot{\delta}_x \\ \delta_v \\ \delta_e
	\end{bmatrix}=
	\begin{bmatrix}
		A_x(\rho) & B_{xw}(\rho) & B_{xd}(\rho) \\
		C_v(\rho) & D_{vw}(\rho) & D_{vd}(\rho) \\
		C_e(\rho) & D_{ew}(\rho) & D_{ed}(\rho)
	\end{bmatrix}
	\begin{bmatrix}
		\delta_x \\ \delta_w \\ \delta_d
	\end{bmatrix}
\end{equation}
where $ A_x=\frac{\partial f}{\partial x} $, $ B_w=\frac{\partial f}{\partial w} $, $ B_d=\frac{\partial f}{\partial d} $, $ C_v=\frac{\partial g}{\partial x} $, $ D_{vw}=\frac{\partial g}{\partial w} $, $ D_{vd}=\frac{\partial g}{\partial d} $, $ C_e=\frac{\partial h}{\partial x} $, $ D_{ew}=\frac{\partial h}{\partial w} $ and $ D_{ed}=\frac{\partial h}{\partial d} $. Let $ (\Psi_k,M_k) $ be a factorization for $ \Pi_k $ and $ \Psi $ be the aggregated system of $ \{\Psi_k\}_{1\leq k\leq N} $ which can yield a (minimal) state-space realization with differential dynamics as follows:
\begin{equation}\label{eq:delta-psi}
	\begin{bmatrix}
		\dot{\delta}_\psi \\ \delta_{z_k}
	\end{bmatrix}=
	\begin{bmatrix}
		A_\psi & B_{\psi v} & B_{\psi w} \\
		C_{z_k} & D_{z_k v} & D_{z_k w}
	\end{bmatrix}
	\begin{bmatrix}
		\delta_\psi \\ \delta_v \\ \delta_w
	\end{bmatrix},\; 1\leq k\leq N
\end{equation}
where $ \delta_\psi\in\R^{n_\psi} $ is the filter state, $ \delta_{z_k} $ is the output of the filter $ \Psi_k $ driven by the signals $ (\delta_v,\delta_w) $. From \eqref{eq:delta-G} and \eqref{eq:delta-psi}, we can obtain the extended system of $ \delta G $ and $ \Psi $ as follows
\begin{equation}\label{eq:extend-system-1}
	\begin{bmatrix}
	\dot{\delta}_\chi \\ \delta_{z_k} \\ \delta_e
	\end{bmatrix}=
	\begin{bmatrix}
	\A(\rho) & \B_{w}(\rho) & \B_{d}(\rho) \\
	\C_{z_k}(\rho) & \D_{z_kw}(\rho) & \D_{z_kd}(\rho) \\
	\C_e(\rho) & \D_{ew}(\rho) & \D_{ed}(\rho)
	\end{bmatrix}
	\begin{bmatrix}
	\delta_\chi \\ \delta_w \\ \delta_d
	\end{bmatrix},\; 1\leq k\leq N
\end{equation}
where $ \chi=\col(x,\psi)\in\R^{n_x+n_\psi} $. The extended system can be expressed in terms of the state matrices in \eqref{eq:delta-G} and \eqref{eq:delta-psi}.

The following result establishes a pointwise LMI condition for differential dissipativity \eqref{eq:diff-dissipation} of the closed-loop system, and can be seen as an application of the method of \cite{Pfifer:2016} to the differential dynamics \eqref{eq:extend-system-1}.

\begin{prop}\label{prop:differential-L2-gain}
	Let the nominal system $ G $ and the uncertainty $ \varDelta $ satisfy Assumption~\ref{asmp:2}-\ref{asmp:4}. The feedback interconnection $ \Fc_u(G,\varDelta) $ has a differential $ L_2 $-gain bound of $ \alpha $ if there exists a smooth state-dependent matrix function $ P(x)=P'(x) $ and a multiplier coefficient vector $ \lambda\in\Lambda $ such that the following pointwise LMI (omitting $ \rho $ for a shorter notation)
	\begin{equation}\label{eq:LMI-soft}
		\begin{split}
		&
		\begin{bmatrix}
			P\A+\A'P+\dot{P} & P\B_w & P\B_d \\
			\B_w'P & 0 & 0 \\
			\B_d'P & 0 & -\alpha^2I
		\end{bmatrix}+
		\begin{bmatrix}
		\C_e' \\ \D_{ew}' \\ \D_{ed}'
		\end{bmatrix}
			\begin{bmatrix}
		\C_e' \\ \D_{ew}' \\ \D_{ed}'
		\end{bmatrix}' \\
		&\quad +\sum_{k=1}^{N}\lambda_k
		\begin{bmatrix}
		\C_{z_k}' \\ \D_{z_kw}' \\ \D_{z_kd}'
		\end{bmatrix}M_k
		\begin{bmatrix}
		\C_{z_k}' \\ \D_{z_kw}' \\ \D_{z_kd}'
		\end{bmatrix}'<0
		\end{split} 
	\end{equation}
	holds for all $ \rho $. 
\end{prop}

Note that if the nominal system is polynomial, the above pointwise LMI condition can be efficiently solved for using SOS programming \cite{Parrilo:2003}.



The main part of the proof parallels Theorem 2 in \cite{Pfifer:2016}, since the evaluation of the differential dynamics \eqref{eq:delta-G} along any particular trajectory gives an LPV system. With this in mind, we simply sketch the proof here.

The key step is to prove that \eqref{eq:LMI-soft} is equivalent to a new LMI formulation which involves a single, hard $ \delta $-IQC and a new matrix function $ \tilde{P}(\chi)\geq 0 $ for all $ \chi $. First, the combined multiplier $ \Pi_\lambda $ has a factorization $ (\Psi_\lambda,M_\lambda) $ where $ \Psi_\lambda=\begin{bmatrix}(sI-A_\psi)^{-1}B_\psi \\ I\end{bmatrix} $ with $ B_\psi:=\begin{bmatrix} B_{\psi v} & B_{\psi w}\end{bmatrix} $ and 
\begin{equation}
M_\lambda=\begin{bmatrix}
Q_\lambda & S_\lambda \\ S_\lambda' & R_\lambda 
\end{bmatrix}:=\sum_{k=1}^{N}\lambda_k
\begin{bmatrix}
\C_{z_k}' \\ \D_{z_kw}' \\ \D_{z_kd}'
\end{bmatrix}M_k
\begin{bmatrix}
\C_{z_k}' \\ \D_{z_kw}' \\ \D_{z_kd}'
\end{bmatrix}'
\end{equation}
with $ Q_\lambda=Q_\lambda' $ and $ R_\lambda=R_\lambda' $. 
From Assumption~\ref{asmp:3}-\ref{asmp:4} and the definition of $ \Lambda $, we have $ \Pi_{\lambda,vv}(j\omega)>0 $ and $ \Pi_{\lambda,ww}(j\omega)<0 $ for all $ \omega\in\R\cup\{\infty\} $. The multiplier $ \Pi_{\lambda} $ can yield a $ J $-spectral factorization $ (\widetilde{\Psi}_\lambda,\widetilde{M}_\lambda) $ with $ \widetilde{M}_\lambda=\diag(I_{n_v},-I_{n_w}) $ \cite[Lemma 4]{Seiler:2015}. The differential dynamics for the state-space realization of $ \widetilde{\Psi}_\lambda $ is
\begin{equation}\label{eq:delta-psi-lambda}
\begin{bmatrix}
\dot{\delta}_\psi \\ \delta_{z_\lambda}
\end{bmatrix}=
\begin{bmatrix}
A_\psi & B_{\psi v} & B_{\psi w} \\
{C}_{z_\lambda} & {D}_{z_\lambda v} & {D}_{z_\lambda w}
\end{bmatrix}
\begin{bmatrix}
\delta_\psi \\ \delta_v \\ \delta_w
\end{bmatrix}
\end{equation}
where $ C_{z_\lambda}=\widetilde{M}_\lambda (D_{z_\lambda}^{-1})'(B_\psi'X+S_\lambda) $ with $ X $ as the stabling solution to the $ ARE(A_\psi,B_\psi,Q_\lambda,S_\lambda,R_\lambda) $ and $ D_{z_\lambda}:=\begin{bmatrix}{D}_{z_\lambda v} & {D}_{z_\lambda w}\end{bmatrix} $ satisfying $ R_\lambda=D_{z_\lambda}'\widetilde{M}_\lambda D_{z_\lambda} $. Then the extended system of $ \delta G $ and $ \widetilde{\Psi}_\lambda $ can be represented by 
\begin{equation}\label{eq:diff-dynamics-closedloop}
\begin{bmatrix}
\dot{\delta}_\chi \\ \delta_{z_\lambda} \\ \delta_e
\end{bmatrix}=
\begin{bmatrix}
\widetilde{\A} & \widetilde{\B}_{w} & \widetilde{\B}_{d} \\
\widetilde{\C}_{z_\lambda} & \widetilde{\D}_{z_\lambda w} & \widetilde{\D}_{z_\lambda d} \\
\C_e & \D_{ew} & \D_{ed}
\end{bmatrix}
\begin{bmatrix}
\delta_\chi \\ \delta_w \\ \delta_d
\end{bmatrix}.
\end{equation}
The equivalent formulation to LMI \eqref{eq:LMI-soft} can be written as:
\begin{equation}\label{eq:LMI-hard}
\begin{split}
&
\begin{bmatrix}
\widetilde{P}\widetilde{\A}+\widetilde{\A}'\widetilde{P}+\dot{\widetilde{P}} & \widetilde{P}\widetilde{\B}_w & \widetilde{P}\widetilde{\B}_d \\
\widetilde{\B}_w'\widetilde{P} & 0 & 0 \\
\widetilde{\B}_d'\widetilde{P} & 0 & -\alpha^2I
\end{bmatrix}+
\begin{bmatrix}
\C_e' \\ \D_{ew}' \\ \D_{ed}'
\end{bmatrix}
\begin{bmatrix}
\C_e' \\ \D_{ew}' \\ \D_{ed}'
\end{bmatrix}' \\
&\quad +
\begin{bmatrix}
\widetilde{\C}_{z_\lambda}' \\ \widetilde{\D}_{z_\lambda w}' \\ \widetilde{\D}_{z_\lambda d}'
\end{bmatrix}\widetilde{M}_\lambda
\begin{bmatrix}
\widetilde{\C}_{z_\lambda}' \\ \widetilde{\D}_{z_\lambda w}' \\ \widetilde{\D}_{z_\lambda d}'
\end{bmatrix}'<0
\end{split} 
\end{equation}
where $ \widetilde{P}=P-\begin{bmatrix}0 & 0 \\ 0 & X\end{bmatrix}\geq 0 $ with $ X $ as the stabilizing solution to the $ ARE(A_\psi,B_\psi,Q_\lambda,S_\lambda,R_\lambda) $.
Left- and right-multiplying \eqref{eq:LMI-hard} with $ \col(\delta_\chi,\delta_w,\delta_d) $ and its transpose, and taking integration of $ t $ over $ [0,T] $ gives
\begin{equation}\label{eq:diff-gain}
\widetilde{V}_T-\widetilde{V}_0+\int_{0}^{T}\delta_{z_\lambda}'\widetilde{M}_\lambda\delta_{z_\lambda}\leq \int_{0}^{T}(\alpha^2|\delta_d|^2-|\delta_e|^2)dt
\end{equation}
where  $ \widetilde{V}_t $ denotes $\delta_\chi'(t)\Mc(\chi(t))\delta_\chi(t) $ with  $ \Mc=\widetilde{P}+\epsilon I $. Here $ \epsilon $ is a small positive constant. Note that $ (\widetilde{\Psi}_\lambda, \widetilde{M}_\lambda) $ is a hard factorization which implies $ \int_{0}^{T}\delta_{z_\lambda}'\widetilde{M}_\lambda\delta_{z_\lambda}\geq 0 $ for all $ T>0 $. Thus, \eqref{eq:diff-gain} is a differential dissipation condition which yields a differential $ L_2 $-gain bound of $ \alpha $ for $ \mathcal{F}_u(G,\varDelta) $. 


There are two major benefits for applying contraction analysis to nonlinear systems: local differential stability implies global incremental/global stability, and no knowledge about the reference trajectory is required. The following theorem shows that robust contraction analysis based on $ \delta $-IQC preserves these two features.
\begin{thm}\label{thm:main}
	Suppose that conditions for Proposition~\ref{prop:differential-L2-gain} are satisfied. If differential dissipation inequality \eqref{eq:diff-gain} holds for any solution $ \rho(\cdot) $ and any solution family $ \Omega_\rho $, then $ \Fc_u(G,\varDelta) $ has an incremental $ L_2 $-gain bound of $ \alpha $. If \eqref{eq:diff-gain} only holds for all $ \Omega_{\rho^*} $ with certain $ \rho^*(\cdot) $, then $ \Fc_u(G,\varDelta) $ has a global $ L_2 $-gain bound of $ \alpha $.  
\end{thm}

\begin{proof}
 We prove the first claim that if \eqref{eq:diff-dissipation} is satisfied for all solution families, an incremental $ L_2 $-gain bound is guaranteed. For any pair of solutions $ \rho_0(\cdot) $ and $ \rho_1(\cdot) $, we consider an one-parameter solution family $ \Omega_{\rho_0} $ satisfying 
\begin{equation}
	\begin{split}
	\overline{x}(0,s)&=c(s,x_0(0),x_1(0)) \\
	\overline{d}(t,s)&=d_0(t)(1-s)+d_1(t)s \\
	\overline{e}(t,s)&=h(\overline{x},\overline{d})(t,s)
	\end{split}
\end{equation}
where $ c(\cdot,x_0(0),x_1(0)) $ is a smooth path joining $ x_0(0) $ and $ x_1(0) $. Thus, $ \Omega_{\rho_0} $ also satisfies $ \overline{\rho}(\cdot,1)=\rho_1(\cdot) $.
Substituting $ (\delta_x,\delta_d,\delta_e)=(\overline{x}_s,\overline{d}_s,\overline{e}_s) $ into \eqref{eq:diff-gain} and integrating it over $ [0,1] $ yield
\begin{equation}
\int_{0}^{T}\int_{0}^{1}|\overline{e}_s|^2dsdt\leq \alpha^2\int_{0}^{T}\int_{0}^{1}|\overline{d}_s|^2dsdt+\int_{0}^{1}c_s'\Mc c_sds.
\end{equation}
This gives the incremental $ L_2 $-gain condition \eqref{eq:global-L2} with $ b(x_0,x_1)=\ell^2(c(\cdot,x_0(0),x_1(0))) $ since $ \overline{d}_s=d_1-d_0 $ and $ |e_1-e_0|^2=|\int_{0}^{1}\overline{e}_sds|^2\leq \int_{0}^{1}|\overline{e}_s|^2ds $ (Cauchy-Schwarz inequality). For the second claim, it is straight forward by following the above steps.
\end{proof}

\section{Illustrative Example}\label{sec:example}
A simplified model of surge-stall dynamics of a jet engine has the form of
\begin{equation}
	\begin{bmatrix}
	\dot{\psi} \\ \dot{\phi}
	\end{bmatrix}=f(x)+Bu+Ed:=
	\begin{bmatrix}
	\phi+u \\ -\psi-\frac{3}{2}\phi^2-\frac{1}{2}\phi^3+d
	\end{bmatrix}
\end{equation}
where $ x=(\phi,\psi) $ is the state, $ u $ is the control input, and $ d $ is external disturbance. Here $ \phi $ is a measure of mass flow through the compressor, and $ \psi $ is a measure of pressure rise in the compressor. By implementing the CCM based control synthesis approach (\cite{Manchester:2018}), we found a constant Riemannian metric and a differential controller 
\begin{equation}
	\delta_u=K(x)\delta_x
\end{equation}
which achieves a bounded global $ L_2 $ gain of $ \alpha=0.93 $ from $ d $ to $ e:=\phi+0.1u $ within the region $ |\phi|\leq 1 $. The control realization based on integration along geodesics is
\begin{equation}
	u(t)=\kappa(x,x^*,u^*)(t):=u^*(t)+\int_{0}^{1}K(\gamma(t,s))\gamma_sds
\end{equation}
where $ (x^*,u^*)(\cdot) $ is a reference trajectory, $ \gamma(t,\cdot) $ is the geodesic joining $ x^*(t) $ to $ x(t) $. For constant metric, there exist a unique geodesic $ \gamma(t,s)=x^*(t)(1-s)+x(t)s $.

The $ \delta $-IQC based contraction analysis problem we considered here is the performance degradation caused by uncertain input delays:
\begin{equation}
	u(t)=\kappa(x,x^*,u^*)(t-\theta)
\end{equation}
where $ \theta\in[0,\Theta] $ with $ \Theta $ as a known bound. The perturbed nonlinear system can be represented as a feedback interconnection of a nominal model $ G $ and a perturbation $ \varDelta $:
\begin{equation}
	\begin{split}
	G:\; &
	\begin{cases}
	\dot{x}&=f(x)+Bk(x,x^*,u^*)+Bw+Ed \\
	e&=Cx+Dk(x,x^*,u^*)+Dw \\
	v&=k(x,x^*,u^*) 
	\end{cases} \\
	\varDelta:\; & \quad w=v(t-\theta)-v(t).
	\end{split}
\end{equation}
Since the uncertainty $ \varDelta $ is a linear infinite-dimensional system, the differential operator $ \partial\varDelta $ exists and shares the same multipliers for $ \delta $-IQCs and conventional IQCs (Corollary~\ref{coro:2}). We can obtain a simple (and not complete) set of $ \delta $-IQCs from \cite{Megretski:1997}:
\begin{equation}\label{eq:iqc}
\begin{split}
\left|\hat{\delta}_v(j\omega)\right|^2-\left|\hat{\delta}_v(j\omega)+\hat{\delta}_w(j\omega)\right|^2&\geq 0 \\
\eta(\Theta\omega)\left|\hat{\delta}_v(j\omega)\right|^2-\left|\hat{\delta}_w(j\omega)\right|^2&\geq 0
\end{split}
\end{equation}
where 
\begin{equation}
\eta(\omega)=\frac{\omega^2+0.08\omega^4}{1+0.13\omega^2+0.02\omega^4}.
\end{equation}
Solving the parameter-dependent LMI \eqref{eq:LMI-soft} problem took using SOS programming (Yalmip \cite{Lofberg:2004} and SDP solver Mosek). The results are shown in Table~\ref{tab:1}. System performance deterioration is observed as the input delay increases.

\begin{table}[!bt]
	\caption{Global $ L_2 $-gain bound $\alpha$ vs delay bound  $ \Theta $ }\label{tab:1}
	\centering
	\begin{tabular}{|c|c|c|c|c|c|}
		\hline \hline
		$ \Theta $ & 0 & 0.04 & 0.08 & 0.12 & 0.16 \\
		\hline
		$ \alpha $ & 0.93 & 1.44 & 4.87 & 13.59 & 546.3 \\ 
		\hline
	\end{tabular}
\end{table}

\section{Conclusion}
This paper extended the time-domain IQC theorem to the contraction analysis framework. This is quite a natural development, since contraction is based on the study of differential dynamics, which can be interpreted as a special type of LPV system. However, by integrating along paths in state-space we obtain rigorous global results for a nonlinear uncertain system.  Our approach leads to a computationally tractable condition to assess the robust contraction performance of a nominal nonlinear system interconnected with uncertainties described by differential IQCs. Future work will consider the synthesis of robust controllers for uncertain nonlinear systems using the approach of \cite{Manchester:2018}.



\bibliographystyle{IEEEtran}
\bibliography{ref}
\end{document}